\documentclass[letterpaper]{article}%
\usepackage{amsfonts,amsmath,amsthm,amssymb,mathtools,dsfont}
\usepackage{geometry,color}
\usepackage{graphicx}
\usepackage{adjustbox,float,array,subcaption}
\usepackage{booktabs,appendix,multirow}

\newcommand{\E}{\mathbb{E}\thinspace }
\newcommand{\e}{\mathcal{E}\thinspace }
\newcommand{\D}{\mathfrak{D}}
\newcommand{\R}{\mathbb{R}}
\newcommand{\PP}{\mathbb{P}\thinspace }

\newtheorem{theorem}{Theorem}

\newtheorem{assumption}[theorem]{Assumption}

\newtheorem{condition}[theorem]{Condition}

\newtheorem{definition}[theorem]{Definition}

\newtheorem{lemma}[theorem]{Lemma}

\newtheorem{proposition}[theorem]{Proposition}
\newtheorem{remark}[theorem]{Remark}

\numberwithin{equation}{section}
\numberwithin{theorem}{section}

\begin{document}

\title{On the Market-Neutrality of Optimal Pairs-Trading Strategies}

\author{Bahman Angoshtari
\thanks{Department of Mathematics, University of Michigan. Email: \texttt{bango@umich.edu}. This article is part of the author's D.Phil. thesis at University of Oxford. The author would like to thank Anders Rahbek, Ronnie Sircar, Hoi Ying Wong, Thaleia Zariphopoulou, and Xunyu Zhou for their comments and suggestions. Support from the Oxford-Man Institute of Quantitative Finance is acknowledged.}
}

\date{\today}

\maketitle

\begin{abstract}

\noindent We consider the problem of optimal investment in a market with two cointegrated stocks and an agent with CRRA utility. We extend the findings of Liu and Timmermann [\emph{The Review of Financial Studies}, 26(4):1048-1086, 2013] by paying special attention to when/if the associated stochastic control problem is well-posed and providing a verification result. Our new findings lead to a sharp well-posedness condition which is, surprisingly, also the necessary and sufficient condition for the optimal investment to be market-neutral (i.e. having offsetting long/short positions in the stocks). Hence, we provide a theoretical justification for market-neutral pairs-trading which, despite having a strong practical relevance, has been lacking a theoretical ground.   \vspace{1ex}

\end{abstract}

\section{Introduction}
This article is a contribution to portfolio management using assets whose price processes are \emph{cointegrated}. Such processes have the property that linear combinations of them is stationary. Intuitively speaking, two cointegrated processes never get too far apart and have a long-run equilibrium with respect to each other. Many economic and financial data series are known to exhibit these properties. Examples include interest rates (\cite{EngleGranger1987} and \cite{Halletal1992}), foreign exchange rates (\cite{BaillieBollerslev1989}), equities (\cite{CerchiHavenner1988}), equity indices (\cite{TaylorTonks1989}), future and spot prices (\cite{BrennerKroner1995}), and commodities (\cite{MacDonaldTaylor1988}).

In portfolio management, there are specific strategies for trading assets which have co-movement in their prices. Such strategies are referred to as ``\emph{pairs trading}'', ``\emph{spread trading}'', or ``\emph{convergence trading}''. These strategies involve identifying two or more assets whose prices are driven by common economic forces, and then trading on any temporary deviation of the prices from their long-run equilibrium. We refer the reader to \cite{Ehrman2006} and \cite{leung2015optimal} for a detailed exposition on pairs-trading as well as on historical insights.

There are two major themes in the convergence trading literature: 1) Empirical studies on profitability of convergence trading; and 2) Theoretical studies on optimal convergence trading.

The first extensive empirical study on convergence trading was provided by \cite{GatevGoetzmannRouwenhorst1999,GatevGoetzmannRouwenhorst2006} where they documented economically significant profits from implementing a very simple pairs-trading rule in the US equity market over an extended period of time. Other empirical studies in this direction include \cite{PapadakisWysocki2007}, \cite{EngelbergGaoJagannathan2009}, \cite{KhandaniLo2007}, \cite{DoFaff2010}, \cite{AvellanedaLee2010}, \cite{KhandaniLo2007} and \cite{galenkoPopovaPopova2012}. The trading strategies in these empirical studies were all pre-assumed rather than being the outcome of some sort of portfolio optimization. Theoretical studies on convergence-trading in continuous-time optimal portfolio choice settings include \cite{Xiong2001}, \cite{LiuLongstaff2004}, \cite{JurekYang2007}, \cite{Primbsetal2008}, \cite{ChiuWong2011}, \cite{chiu2015dynamic}, \cite{LiuTimmermann2012}, and \cite{tourin2013dynamic}. Assuming that the spread is an Ornstein-Uhlenbeck (O-U) process and that the investors have logarithmic utility, \cite{Xiong2001} formulated a general equilibrium model and solved it numerically. The results showed that pairs-trading can have destabilizing effects on the market. \cite{LiuLongstaff2004} modeled the spread by a Brownian bridge process and provided analytical solution for the associated Merton problem with logarithmic utility. \cite{JurekYang2007} and \cite{Primbsetal2008} considered an O-U spread and solved the optimal expected terminal utility problem for power utilities in closed form. The former study, provided analytical evidence for the potential destabilizing behavior of the convergence traders, consistent with the numerical solution of \cite{Xiong2001} general equilibrium model. Finally, \cite{ChiuWong2011}, \cite{chiu2015dynamic}, \cite{LiuTimmermann2012}, and \cite{tourin2013dynamic} modeled the original cointegrated prices by a \emph{continuous-time error correction model}. \cite{ChiuWong2011, chiu2015dynamic} solved the mean-variance portfolio selection problem, while \cite{LiuTimmermann2012} and \cite{tourin2013dynamic} solved the Merton investment problem with power and exponential utilities, respectively.

All of the empirical and theoretical studies above, apart from \cite{ChiuWong2011},  \cite{chiu2015dynamic}, \cite{LiuTimmermann2012}, and \cite{tourin2013dynamic}, implicitly or explicitly assumed the investor's strategy to be ``\emph{market-neutral}''. When trading equities, market-neutrality is interpreted as offsetting long/short position, such that the (monetary) investments in the assets offset each other at all times. See chapter 2 of \cite{Ehrman2006} for further discussion on different forms of market neutrality and its significance in the practice of convergence-trading.

Despite its widespread use and strong practical relevance, the market-neutrality of the optimal convergence-trading strategy is yet to be justified. Indeed, a rigorous normative study that yields market-neutral policies as optimal has been lacking. To the best of our knowledge, the only study that investigated this assumption is \cite{LiuTimmermann2012} which, however, provides a negative result. Indeed, assuming a market setting in which assets follow a continuous time error correction model and an agent maximizing her power utility of terminal wealth, it is therein shown that the optimal strategies are not market-neutral. Therefore, there is, from a theoretical point of view, an unanswered fundamental question. \emph{How can one justify this investment practice in a theoretical portfolio choice framework? In other words, can one identify a market model and a preference criterion for the investor which support pairs-trading?} The answer to this question will be the main focus of this paper. In other words, the main motivation is to provide \emph{a theoretical ground} for market-neutral pairs-trading, \emph{without a priori restricting the portfolio strategies}.

In this paper, we re-consider the investment model assumed in \cite{LiuTimmermann2012}. Our contribution to the existing literature is threefold. Firstly, we show that the problem might be ill-posed, in the sense that the maximum expected utility becomes infinite in a finite time horizon for specific values of the risk aversion parameter (see Theorem \ref{thm:CRRA}). This finding is consistent with existing literature of optimal investment. In particular, \cite{KimOmberg1996} and \cite{KornKraft2004} has observed this phenomenon and coined the terms ``\emph{nirvana strategies}'' and ``\emph{I-unstable}'' for investment strategies that yield infinite expected utility in finite investment horizon.

Secondly, we provide the so-called ``\emph{verification result}'' for the stochastic control problem (see Theorem \ref{thm:Verify}). This is a necessary step since the problem does not satisfy the Lipschitz or polynomial growth conditions needed for the classical verification result, e.g. Theorem 3.8.1 on page 135 of \cite{FlemingSoner2006}.

Thirdly, we identify the necessary and sufficient condition on market parameters for well-posedness of the investment problem for all values of the investor's risk aversion (see Condition \ref{ch1cond:MNcondition}). Interestingly, the same condition turns out to be the necessary and sufficient condition for market-neutrality of the optimal trading strategy (see Theorem \ref{ch1thm:MNcondition}).

Our findings provide economic viability for the practice of market-neutral pairs-trading by the following argument. Note that our investment model (as well as the one in \cite{LiuTimmermann2012}) is a ``\emph{partial equilibrium}'' model. In other words, it is a priori assumed that the assets are traded in a market in its equilibrium state. Investors who achieve infinite expected utility in a finite investment horizon, cannot exist in a partial equilibrium model, since they would have pushed the marked out of equilibrium by aggressively investing in the assets. Thus, any combination of model parameters that lead to existence of such agents must be excluded in a partial equilibrium model. Our findings show that imposing such restriction is equivalent to assuming that the optimal investment strategy is market-neutral.

The remainder of this paper is organized as follows. In Section \ref{sec:Market}, we describe the financial model. In Section \ref{sec:MN}, we introduce market-neutral investment strategies and discuss their significance. In Section \ref{sec:Optimal}, we formulate and solve the optimal investment problem for CRRA investors. In particular, we pay special attention to when/if the stochastic control problem is well-posed, and provide the verification result. In Section \ref{sec:WellPosed}, we provide our main result by introducing Condition \ref{ch1cond:MNcondition} and showing that it is the necessary and sufficient condition for well-posedness of the optimal investment problem as well as market-neutrality of the optimal strategy. Section \ref{sec:conclude} concludes the paper.

\section{Market setting}\label{sec:Market}
The market consists of a riskless asset with zero interest rate,\footnote{Assuming zero short rate is not restrictive. If the short-rate is non-zero or even time-varying (but deterministic), one can recover the zero short-rate assumption by using the discounted prices.} and two stocks whose price processes, $S=(S_t^1,S_t^2)_{t\ge 0}$ satisfy the ``\emph{continuous-time error correction model}'', i.e.
\begin{equation}
	\frac{dS_t^1}{S_t^1} = \alpha_1 Z_t dt + \sigma_1dW_t^1,\label{ch1eq:StocksDynamics}
\end{equation}
and
\begin{equation}
\frac{dS_t^2}{S_t^2} = \alpha_2 Z_t dt + \sigma_2\rho
dW_t^1 + \sigma_2 \sqrt{1-\rho^2} \thinspace dW_t^2,\label{ch1eq:StocksDynamics_S2}
\end{equation}
where the ``\emph{log-spread}'' $(Z_{t})_{t\geq0}$ is defined as
\begin{equation}\label{ch1eq:CointegrationRegression}
Z_{t} := \log S_{t}^{1}-c\log S_{t}^{2}
+\frac{1}{2}\left(\sigma_1^2-c\sigma_2^2\right)t.
\end{equation}
Here, $(W_t)_{t\ge0} = (W_t^1,W_t^2)^\top_{t\ge0}$ is a two dimensional standard Brownian motion in a filtered probability space $(\Omega, \mathcal{F}, (\mathcal{F}_t)_{t\ge0}, \mathbb{P})$, where $(\mathcal{F}_t)_{t\ge0}$ is the augmented filtration generated by $(W_t)$. All the coefficients are constant, and further assumption on coefficients will be given below. For future reference, we also note that \eqref{ch1eq:StocksDynamics} and \eqref{ch1eq:StocksDynamics_S2} have the matrix form
\begin{equation}\label{ch1eq:StocksDynamics_mat}
d S_t = \text{diag}(S_t)\left( \alpha Z_t dt + \Sigma dW_t\right),
\end{equation}
where
\begin{equation}\label{ch1eq:Sigma}
\Sigma := \begin{pmatrix}
\sigma_1 & 0\\
\sigma_2 \rho & \sigma_2\sqrt{1-\rho^2}
\end{pmatrix}, \quad \text{ and }\quad
\alpha := \begin{pmatrix}
\alpha_1\\
\alpha_2
\end{pmatrix}.
\end{equation}

The following assumptions on the constant coefficients are standing throughout.

\begin{assumption}\label{ch1asm:StandingAsm}
  (i) $\sigma_{1},\sigma_{2}>0$ and $\left\vert \rho\right\vert <1$; (ii) $\alpha_1 < c \alpha_2$; and (iii) $Z_0:= $ is a Gaussian random variable with mean zero and variance
  \[
  \frac{\sigma_1^2 + c^2 \sigma_2^2 -2 c \rho \sigma_1 \sigma_2}
  {2(c\alpha_2-\alpha_1)},
  \]
  and it is independent of $(W_t)_{t\ge0}$.
\end{assumption}

Assumption \ref{ch1asm:StandingAsm}.(i) is the usual non-degeneracy assumption on diffusion-type market models. The rule of Assumption \ref{ch1asm:StandingAsm}.(ii)--(iii) is to enforce $(Z_t)_{t\ge0}$ to be a stationary Ornstein-Uhlenbeck process, as shown by the following lemma.

\begin{lemma}\label{lem:StationarityOfResidual}
The log-spread $(Z_t)_{t\ge0}$ satisfies
\begin{equation}\label{ch1eq:Z_Dynamics1}
dZ_{t}=-\kappa Z_t dt + \sigma_Z dW_t^Z,
\end{equation}
where
\begin{equation}\label{ch1eq:Z_Dynamics2}
\kappa := c \alpha_2 - \alpha_1>0,\quad \sigma_Z^2 := \sigma_1^2 + c^2 \sigma_2^2
-2 c \rho \sigma_1 \sigma_2,
\end{equation}
and
\begin{equation}\label{ch1eq:W_Z}
W_t^Z := \frac{1}{\sigma_Z}\left\{ \left(\sigma_1 - c \sigma_2 \rho\right)
W_t^1 - c \sigma_2 \sqrt{1-\rho^2} W_t^2 \right\}  .
\end{equation}
In particular, $(Z_t)_{t\geq0}$ is an Ornstein-Uhlenbeck process given by
\begin{equation}\label{ch1eq:Z_explicit}
Z_t = e^{-\kappa t} \left(Z_0 + \sigma_Z \int_0^t{e^{\kappa s}dW_s^Z}\right),
\end{equation}
which is a stationary Gaussian process with $E \left(Z_t\right) = 0$ and
\begin{equation}\label{ch1eq:Z_moments}
\E \left(Z_t Z_s\right) =
\frac{\sigma_Z^2}{2\kappa} e^{-\kappa \vert t-s \vert}, \quad t,s\geq0.
\end{equation}
\end{lemma}
\begin{proof}
Applying It\^o's lemma to find the dynamics of $(Z_t)$ from (\ref{ch1eq:StocksDynamics}) and
(\ref{ch1eq:StocksDynamics_S2}) yields (\ref{ch1eq:Z_Dynamics1})-(\ref{ch1eq:W_Z}). \eqref{ch1eq:Z_explicit} is the well-known strong solution of the linear stochastic differential equation \eqref{ch1eq:Z_Dynamics1}, c.f. example 6.8 on p.~358 of \cite{KartzasShreve1991}. The rest of the proof readily follows from (\ref{ch1eq:Z_explicit}). In particular, we need Assumption \ref{ch1asm:StandingAsm}.(ii)--(iii) to obtain $E \left(Z_t\right) = 0$ and \eqref{ch1eq:Z_moments} which, in turn, yield stationarity of $(Z_t)$.
\end{proof}

\begin{remark}
	The proof of Lemma \ref{lem:StationarityOfResidual} reveals that the term $(\sigma_1^2-c\sigma_2^2)t/2$ in (\ref{ch1eq:CointegrationRegression}) is specifically chosen so that $\E Z_t = 0$, which serves three purposes. Firstly, it simplifies the algebra in the proof of several results below. Secondly, from an economic point of view, this assumption means that the long-run equilibrium level of $(Z_t)$ is $0$. From \eqref{ch1eq:StocksDynamics} and \eqref{ch1eq:StocksDynamics_S2}, it then follows that the long-run equity risk premiums are zero and the only reason for investing in the stocks is to exploit short-term deviations from equilibrium when $Z_t\ne0$. Thus, choosing $\E Z_t = 0$ facilitates the analysis of pairs-trading by isolating the effect of cointegration. Finally, estimation methods that are used to calibrate \eqref{ch1eq:StocksDynamics}--\eqref{ch1eq:CointegrationRegression} (e.g. Engle-Granger two-step method and Johansen's approach) alway impose $\E Z_t = 0$.
\end{remark}

Recall that two stochastic processes are \emph{cointegrated} if a linear combination of them is stationary. Hence, Lemma \ref{lem:StationarityOfResidual} implies that the stock log-prices are cointegrated. As it was mentioned in the introduction, more can be said about the connection of the market model considered herein and the theory of cointegration. Indeed, as shown in \cite{KesslerRahbek2001}, \cite{KesslerRahbek2004} and \cite{DuanPliska2004}, the price dynamics given by (\ref{ch1eq:StocksDynamics}) and (\ref{ch1eq:StocksDynamics_S2}) is the diffusion limit of a so-called \emph{error correction model}. These models are discrete-time representations of systems of cointegrated processes. Refer to \cite{Hamilton1994}, \cite{Johansen1995}, and \cite{Juselius2006} for a more detailed exposition on cointegration.

We consider an agent who invests in the market over a fixed trading horizon $[0,T]$ and with an initial endowment $x>0$. An \emph{admissible strategy} $\pi^\top=(\pi^1_t,\pi^2_t)_{t\in[0,T]}$ is defined as an $(\mathcal{F}_t)$-adapted process satisfying
\begin{equation}\label{ch1eq:integrability}
\int_0^T{\left(|\pi_t^\top \alpha Z_t| + \pi_t^\top
\Sigma\Sigma^\top \pi_t \right)dt}<+\infty, \quad \mathbb{P}\text{-almost
surely}.
\end{equation}
Here, $\pi^i_t$ is the portfolio weight of the i-th stock, i.e. the \emph{proportion} of agent's wealth invested in the i-th stock at $t$. Thus, $(1-\pi^1_t-\pi^2_t)_{t\in[0,T]}$ is the proportion of wealth invested in the bank account. Short-selling of the stocks and the bank account is allowed, and is represented by negative portfolio weights.  The set of admissible strategies is denoted by $\mathcal{A}$.

For any admissible strategy $\pi^\top=(\pi^1,\pi^2)\in\mathcal{A}$, the agent's wealth $X^\pi=(X_t^\pi)_{t\in[0,T]}$ is given by the budget constraint
\begin{equation}\label{ch1eq:wealthSE}
	X^\pi_t = x + \int_0^t X^\pi_u Z_u \pi_u^\top \alpha du + \int_0^{t} X^\pi_u \pi^\top_u \Sigma dW_u.
\end{equation}
In particular, $X^\pi_t>0$ $\mathbb{P}$-almost surely for all $t\in[0,T]$, since it is the stochastic exponential of a continuous process.

\section{Market-neutral investment}\label{sec:MN}
As mentioned earlier, the main motivation of this paper is to provide a theoretical ground for the so-called ``\emph{market-neutral}'' trading strategies. In this section, we introduce these strategies and discuss their significance.

Let $\mathcal{F}^Z_t$, $t\ge0$, be the augmentation of $\sigma(Z_u: 0\le u\le t)=\sigma(W^Z_u: 0\le u\le t)$. Note that $\mathcal{F}^Z_t\subset\mathcal{F}_t$, $t\ge0$,
that is, we gain less information from knowing the log-spread $(Z_u)_{0\le u\le t}$ than from knowing the stock prices $(S_u)_{0\le u\le t}$.

The subfiltration $(\mathcal{F}^Z_t)$ has an important role in practice. Since $(Z_t)$ is stationary, it is possible to calibrate a model for $(Z_t)$ using historical data. On the other hand, the stock prices $S^1_t$ and $S^2_t$ are generally not stationary. Thus, calibrated models for $S^1_t$ and $S^2_t$ are less reliable than their counterparts for $Z_t$. It is then natural that practitioners focus on trading strategies and portfolios that are $(\mathcal{F}^Z_t)$-adapted. Such strategies are called ``\emph{market-neutral}'', since their dynamics only depends on a stationary signal $(Z_t)$ which, by design, is ``immune'' to non-stationarity of the market (e.g. bull/bear states of the market).

\begin{definition}\label{def:MN}
	An admissible strategy $\pi^\top=(\pi^1,\pi^2)\in\mathcal{A}$ is market-neutral (M-N) if both $\pi$ and $X^\pi$ are $(\mathcal{F}_t^Z)$-adapted.
\end{definition}

The majority of the existing literature define M-N strategies in a different way. In particular, such strategies are defined as dollar-neutral, share-neutral or beta-neutral, depending on how the cointegration relationships among the prices are defined. For example, when the mean reverting signal is the logarithm of the price differences, i.e. when $c=1$ as in \cite{LiuTimmermann2012}, market-neutrality is interpreted as dollar-neutrality, which requires offsetting long/short positions in the stocks such that the (monetary) investments in the two stocks is zero at all times. This assumption is most common when trading equities. When the mean reverting signal is the price difference, market neutrality is interpreted as share-neutrality, in which case the number of shares (or contracts) in different assets offset each other. This assumption is relevant to futures markets, or when the assets are almost identical. See chapter 2 of \cite{Ehrman2006} for further discussion on different types of market-neutrality and their significance in the practice of convergence-trading.

We consider the following alternative definition of M-N portfolios, which is consistent with the existing literature.

\begin{definition}\label{ch1def:PairsTrading}
	An admissible strategy $\pi=(\pi^1,\pi^2)\in\mathcal{A}$ is M-N if $\pi$ is $(\mathcal{F}_t^Z)$-adapted and
	\begin{equation}\label{ch1eq:MN}
	\pi_{t}^{2}=-c\thinspace\pi_{t}^{1};\quad \PP .a.s.\; \forall t\ge0,
	\end{equation}
	with $c$ as in \eqref{ch1eq:CointegrationRegression}.
\end{definition}

Our next result shows that the two definitions are equivalent.

\begin{lemma}\label{lem:MN}
	Let $\pi=(\pi^1,\pi^2)$ be an $(\mathcal{F}_t^Z)$-adapted admissible strategy. Then, $X^\pi$ is $(\mathcal{F}_t^Z)$-adapted if and only if \eqref{ch1eq:MN} holds.
\end{lemma}
\begin{proof}
	From \eqref{ch1eq:wealthSE}, we have
	\[
		\frac{d X^\pi_t}{X^\pi_t} = \pi^1_t \frac{d S^1_t}{S^1_t} + \pi^2_t \frac{d S^2_t}{S^2_t}
		= \pi^1_t \left(\frac{d S^1_t}{S^1_t} - c \frac{d S^2_t}{S^2_t}\right) + (\pi^2_t+c\pi^1_t) \frac{d S^2_t}{S^2_t}
		= \pi^1_t dZ_t + (\pi^2_t+c\pi^1_t) \frac{d S^2_t}{S^2_t}.
	\]
	Clearly, $(X^\pi_t)$ is $(\mathcal{F}_t^Z)$-adapted if and only if the second term on the right side vanishes, which is equivalent to \eqref{ch1eq:MN}.
\end{proof}

As mentioned earlier, it is a common industry practice to consider M-N investment in pairs-trading scenario. It is then possible to only consider (and calibrate) the dynamics of the log-spread, and ignore the individual asset prices altogether. This reduces the dimensionality of the problem and, more importantly, facilitates the process of model estimation and calibration, as the spread is a stationary process while the original price processes are not.

On the other hand, the approach taken by academics is not uniform. Early studies such as \cite{GatevGoetzmannRouwenhorst2006}, \cite{LiuLongstaff2004}, \cite{JurekYang2007}, and \cite{Primbsetal2008} followed the industry practice by assuming, a priori, that the investment strategies are M-N. On the other hand, more recent studies such as \cite{ChiuWong2011}, \cite{chiu2015dynamic}, \cite{LiuTimmermann2012} and \cite{tourin2013dynamic} do not impose such restriction and find that the optimal strategy is, in general, not M-N. Therefore, there is, from a theoretical point of view, an unanswered fundamental question. \emph{How can one justify investment practice of pairs-trading in a theoretical portfolio choice framework? In other words, can one identify a market model and a preference criterion for the investor which support M-N pairs-trading?} Answering these question will be the main goal of this paper.

\section{Optimal investment problem}\label{sec:Optimal}

In this section, we consider the Merton investment problem in the market setting of Section \ref{sec:Market}, which is one of the problems considered in \cite{LiuTimmermann2012}. Our contribution to the existing literature is as follows.
\begin{itemize}
	\item[(i)] We show that the problem might be ill-posed, in the sense that the maximum expected utility becomes infinite in a finite time horizon for specific values of the risk aversion parameter (i.e. $\gamma\in(0,\gamma_0)$ below).
	
	\item[(ii)] We provide the verification step for our stochastic control problem (see Theorem \ref{thm:Verify}).
\end{itemize}
All of these results are missing from \cite{LiuTimmermann2012}. As will be explained in the next section, item (i) is crucial in achieving our main goal of justifying the M-N pairs-trading practice.

\begin{remark}
	Achieving infinite expected utility has been observed in the context of trading a stock with mean-reverting return. Examples of such studies include \cite{KimOmberg1996} and \cite{KornKraft2004} who, respectively, coined the terms ``\emph{nirvana strategies}'' and ``\emph{I-unstable}'' for investment strategies that yield infinite expected utility in finite investment horizon. 
\end{remark}

We assume that the investor faces the following problem
\begin{equation}\label{eq:Merton}
\underset{\pi\in\mathcal{A}}{\sup}\thinspace\E \frac{(X^\pi_T)^{1-\gamma} -1}{1-\gamma},
\end{equation}
for a constant $0<\gamma<1$.
\begin{remark}
	We only consider CRRA utility with ``\emph{relative risk aversion}'' $\gamma$ in the interval $(0,1)$, i.e. when the agent is more risk seeking than a log-utility investor. As we will see below, the logarithmic case as well as power utilities with $\gamma >1$ are well-posed and, thus, these cases has already been accounted for in \cite{LiuTimmermann2012}.	
\end{remark}
%

Our main insight of this section is identifying the exact \emph{well-posedness} conditions for \eqref{eq:Merton}. In particular, we introduce the ``\emph{critical relative risk aversion}'',
\begin{equation}\label{ch1eq:gamma0}
\gamma_0:= 1-\left(
\frac{\kappa} {\|(1,-c)\Sigma\| \thinspace \|\Sigma^{-1}\alpha\|}
\right)^2.
\end{equation}
Not that $0\le \gamma_0<1$ since, by the Cauchy-Schwarz inequality,
\[
	\kappa=(-1,c)\alpha \le \|(-1,c)\Sigma\| \thinspace \|\Sigma^{-1}\alpha\|.
\]
We show the following dichotomy.
\begin{enumerate}
	\item [(a)] If $0<\gamma<\gamma_{0}$, then, the Merton problem is \emph{ill-posed}. In particular, the agent's maximal expected utility of wealth increases rapidly with the investment horizon $T$ and approaches infinity at a finite critical horizon $T_N(\gamma)>0$, which is explicitly given by (\ref{eq:TN}). See Theorem \ref{thm:CRRA} below and, in particular, \eqref{eq:Nirvana}.
	
  \item [(b)] If $\gamma\geq\gamma_{0}$, then the Merton problem is \emph{well-posed}, in the sense that the value function is finite, for \emph{any choice} of time horizon $T>0$.
\end{enumerate}
\begin{remark}
	Note that since $\gamma_0 < 1$, the logarithmic case as well as power utilities with $\gamma >1$ are always well-posed. The ill-posed case may only happen for power utilities with $0<\gamma<1$, i.e. when the agent is more \emph{risk seeking} than a log-utility investor.	
\end{remark}

For the rest of this section, we solve the investment problem \eqref{eq:Merton} and show that the aforementioned dichotomy (a)-(b) holds. The value function corresponding to the stochastic control problem \eqref{eq:Merton} is given by,
\begin{equation}\label{eq:VF}
v(x,z,t) = \underset{\pi\in\mathcal{A}}{\sup}\thinspace
\E^{x,z,t}\frac{(X^\pi_T)^{1-\gamma} -1}{1-\gamma};\quad x>0, z\in\R, t\in[0,T].
\end{equation}
Here, $\E^{x, z, t}$ indicates that we condition the expectation on $X_t = x$ and $Z_t = z$. Theorem \ref{thm:CRRA} solves the related Hamilton-Jacobi-Bellman (HJB) equation and Theorem \ref{thm:Verify} provides the so-called verification result, i.e. that the solution of the HJB equation is indeed the value function.

The HJB equation corresponding to the value function \eqref{eq:VF} is
\begin{equation}\label{eq:HJB}
\begin{cases}
	\displaystyle \sup_{\pi\in\R^2} \{\mathcal{L}^\pi v(x,z,t)\}=0,\\
	v(x,z,T)=\frac{x^{1-\gamma} -1}{1-\gamma},
\end{cases}
\end{equation}
for $(x,z,t)\in\R^+\times \R \times [0,T)$. Here, the differential operator $\mathcal{L}^\pi$ is given by
\begin{equation}\label{eq:Kol}
	\mathcal{L}^\pi f := f_t - \kappa z f_z + \frac{1}{2} \sigma_Z^2 f_{zz} + x z \alpha^\top \pi f_x +\frac{1}{2} x^2 \pi^\top \Sigma\Sigma^\top \pi f_{xx} + x (1,-c) \Sigma\Sigma^\top \pi f_{xz},
\end{equation}
in which $\pi\in\R^2$ and $f$ is assumed to be twice differentiable with respect to $x$ and $z$ and differentiable with respect to $t$.

Theorem \ref{thm:CRRA} provides the solution to the HJB equation. To state the result, we need the following definitions. The ``\emph{critical time horizon}'' $T_N(\gamma)\in(0,+\infty]$ is given by
\begin{equation}\label{eq:TN}
	T_N(\gamma)=
	\begin{cases}
		+\infty; & \gamma\ge\gamma_0,\\
		\frac{\gamma}{\sigma_Z \|\Sigma^{-1}\alpha\|\sqrt{\gamma_0-\gamma}} \bigg(
		\frac{\pi}{2} + \arctan \Big( \frac{\kappa \gamma}{\sigma_Z
		\|\Sigma^{-1}\alpha\|\sqrt{\gamma_0-\gamma}} \Big) \bigg); &0<\gamma<\gamma_0.	
	\end{cases}
\end{equation}
We also introduce the ``\emph{discriminant}''
\begin{equation}\label{ch1eq:MetonEscapDisc}
	\D = \frac{\sigma_Z^2 \|\Sigma^{-1}\alpha\| ^2} {\gamma^2} (\gamma-\gamma_0),
\end{equation}
and the functions
\begin{equation}\label{eq:g}
 g(t) = \frac{\kappa}{2\gamma} t - \frac{1}{2}
	\begin{cases}
    \log \Big( \cosh\big( t\sqrt{|\D|} \big) +\frac{\kappa}{\gamma\sqrt{|\D|}} \sinh\big(
    t\sqrt{|\D|} \big) \Big);
    &\text{if }\gamma\ne\gamma_0,\\
    \vphantom{\Bigg()}
    \log \Big( 1 + \frac{\kappa}{\gamma} t \Big) ; &\text{if
    }\gamma=\gamma_0,
	\end{cases}
\end{equation}
and
\begin{equation}\label{eq:h}
    h(t) =
	\begin{cases}
    \frac{\displaystyle\vphantom{\Big(} (1-\gamma)
    \|\Sigma^{-1}\alpha\|^2} {\displaystyle\vphantom{\bigg(} \kappa\gamma +\gamma^2
    \sqrt{\mathfrak{D}} \coth \Big( t\sqrt{\mathfrak{D}} \Big) };
    &\quad \text{if }\gamma_0<\gamma<1,\\
    \frac{\displaystyle\vphantom{\big(} \kappa}
    {\displaystyle\vphantom{\Big(}\gamma \sigma_Z^2}\left(1-
    \frac{\displaystyle\vphantom{\big(}\gamma}
    {\displaystyle\vphantom{\Big(} \gamma + \kappa t}
    \right); &\quad \text{if }\gamma=\gamma_0,\\
	-\frac{\sqrt{-\mathfrak{D}}}
    {\sigma_Z^2}\tan\bigg( \arctan \Big( \frac{\kappa}
    {\gamma\sqrt{-\mathfrak{D}}}\Big) -
    \sqrt{-\mathfrak{D}}t \bigg) + \frac{\kappa}
    {\gamma\sigma_Z^2};
	&\quad \text{if }0<\gamma<\gamma_0.\\
	\end{cases}
\end{equation}

\vskip1em

\begin{theorem}\label{thm:CRRA}
	For $T< T_N(\gamma)$, the solution of the HJB equation \eqref{eq:HJB} is given by
	\begin{equation}\label{eq:VF2}
		v(x,z,t)= \frac{x^{1-\gamma} \left(e^{g(T-t)+\frac{1}{2}h(T-t)z^2}\right)^\gamma - 1}{1-\gamma};\quad (x,z,t)\in\R^+\times\R\times[0,T].
	\end{equation}
	Furthermore, for $(z,t)\in\R\times[0,T]$, the maximizer $\pi$ in \eqref{eq:HJB} is given by
	\begin{equation}\label{eq:optimalPort}
	\pi^\star(z,t) = \Big[\frac{1}{\gamma} (\Sigma\Sigma^\top)^{-1}\alpha + h(T-t)
	\begin{pmatrix}
	 1\\
	 -c
	 \end{pmatrix}\Big] z.
	\end{equation}
	Finally, for $0<\gamma<\gamma_0$, one has
	\begin{equation}\label{eq:Nirvana}
		\lim_{T\to T_N(\gamma)^-} v(x,z,0) = +\infty;\quad \forall (x,z)\in\R^+\times\R.
	\end{equation}
\end{theorem}
\vskip 5pt
\begin{proof}
	Assuming $v_{xx}(x,z,t)\le0$ (which will be verified later) yields that the optimizer in the point-wise optimization problem $\sup_{\pi\in\R^2} \{\mathcal{L}^\pi v(x,z,t)\}$ is
	\begin{equation}\label{ch1eq:Candidate}
	\pi^\star(x,z,t) := -\frac{z\,v_x(x,z,t)}{x\,v_{xx}(x,z,t)} 
	(\Sigma\Sigma^\top)^{-1} \alpha - \frac{v_{xz}(x,z,t)}{x\,v_{xx}(x,z,t)}
	\begin{pmatrix}
		1\\
		-c
	\end{pmatrix}.
	\end{equation} 
	By substituting $\pi^\star$ into \eqref{eq:HJB}, one obtains the fully non-linear Cauchy problem
	\begin{equation}
		\begin{cases}\label{eq:HJB2}
			v_{t}-\kappa z v_{z}+\frac{1}{2}\sigma_{Z}^{2}v_{zz} -\frac{1}{2}\|\Sigma^{-1}\alpha\|^2 z^2 \frac{v_x^2}{v_{xx}} - \frac{1}{2} \sigma_{Z}^{2}\frac{v_{xz}^2}{v_{xx}}
						+ \kappa z \frac{v_x v_{xz}}{v_{xx}}=0,\\
			v(x,z,T)=\frac{x^{1-\gamma} -1}{1-\gamma},
		\end{cases}
	\end{equation}
	for $(x,z,t)\in\R^+\times \R \times [0,T)$. Substituting the ansatz
	\begin{equation}\label{eq:ansatz}
		v(x,z,t) = \frac{\varphi(z,t)^\gamma
	x^{1-\gamma}-1}{1-\gamma}
	\end{equation}
	into \eqref{eq:HJB2} yields that the unknown function $\varphi$ satisfies
	\begin{equation}\label{ch1eq:phiPDE}
	\begin{cases}
	\varphi_t -\frac{1}{\gamma}\kappa z \varphi_z + \frac{1}{2} \sigma_Z^2
	\varphi_{zz} + \frac{1-\gamma}{2\gamma^2} z^2 \|\Sigma^{-1}\alpha\|^2 \varphi = 0;
	&\quad (t,z) \in [0,T)\times \mathbb{R},\\
	\varphi(z,T)=1;&\quad z\in\R.
	\end{cases}
	\end{equation}
	This PDE is solved in Appendix \ref{ch1app:2ndPDE}. In particular, by taking $\mathbf{a} = (1/\gamma) \Sigma^{-1}\alpha$, $\mathbf{b}^\top=(1,-c)\Sigma$, and $\xi=1-\gamma$, one may re-write (\ref{ch1eq:phiPDE}) as (\ref{ch1eq:2ndPDE}). The corresponding escape criterion discriminant defined by (\ref{ch1eq:EscapeDisc}) is
	\[
	\mathfrak{D} = \frac{\kappa^2}{\gamma^2} - \frac{1-\gamma} {\gamma^2} \sigma_Z^2
	\|\Sigma^{-1}\alpha\|^2 = \frac{\sigma_Z^2 \|\Sigma^{-1}\alpha\| ^2} {\gamma^2} (\gamma-\gamma_0),
	\]
	which coincide with \eqref{ch1eq:MetonEscapDisc}. In particular, $\D\ge0$ if and only if $\gamma\ge\gamma_0$. Therefore, $T_{esc}$, $g$, and $h$ of Appendix \ref{ch1app:2ndPDE} become $T_N$ of \eqref{eq:TN}, $g$ of \eqref{eq:g}, and $h$ of \eqref{eq:h}, respectively. Lemma \ref{lem:AuxPDE} then yields that
	\[
		\varphi(z,t) = \exp\Big( g(T-t) + \frac{1}{2} h(T-t) z^2\Big);\quad (z,t)\in\R\times[0,T],
	\]
	and substituting into \eqref{eq:ansatz} yields the solution \eqref{eq:VF2}. It can be easily checked that $v_{xx}\le0$, and \eqref{ch1eq:Candidate} yields the candidate optimal control of \eqref{eq:optimalPort}. Finally, by Lemma \ref{lem:AuxPDE}, it follows that, for all $z>0$, $\lim_{T\to  T_N(\gamma)^-}\varphi(z,0)=+\infty$ which, in turn, yield \eqref{eq:Nirvana}.
\end{proof}

We end this section by providing the so-called verification step. In other words, we show that \eqref{eq:VF2}, i.e. the solution of the HJB equation, is the value function given by \eqref{eq:VF}.

\begin{remark}
	Classical verification results, e.g. Theorem 3.8.1 on page 135 of \cite{FlemingSoner2006}, require either Lipschitz conditions on the coefficients of state equations or polynomial growth of the candidate value function. Neither of these conditions holds in our setting. In particular, the Lipschitz conditions fail because of the term $X^\pi_u Z_u$ in the drift of \eqref{ch1eq:wealthSE} and that $(X^\pi_t)$ and $(Z_t)$ are both unbounded. Furthermore, as the following lemma shows, $h$ in \eqref{eq:VF2} is strictly positive and the candidate value function $v$ has exponential growth in $z$.
	\begin{lemma}\label{lem:Monotone}
		The function $h$ given by \eqref{eq:h} is strictly positive and strictly increasing on $\big[0,T_N(\gamma)\big)$.
	\end{lemma}
	\begin{proof}
		From the proof of Theorem \ref{thm:CRRA}, the function $h$ satisfy the Riccati equation \ref{ch1eq:Riccati}, where $\mathbf{a} = (1/\gamma) \Sigma^{-1}\alpha$, $\mathbf{b}^\top=(1,-c)\Sigma$, and $\xi=1-\gamma$. Since, $\xi=1-\gamma>0$, Lemma \ref{lem:RDE}.(i) yields that $h$ is strictly increasing and positive.
	\end{proof}
\end{remark}

Since requirements of the classical verification results are not satisfied, we provide a verification theorem tailored to our control problem.

\begin{theorem}\label{thm:Verify}
	For $T< T_N(\gamma)$, the function $v$ given by \eqref{eq:VF2} coincides with the value function \eqref{eq:VF}. Furthermore, with slight abuse of notation, the optimal investment strategy is given by $\pi^* = \big(\pi^*(Z_t,t)\big)_{t\in[0,T]}$, with the function $\pi^*(\cdot,\cdot)$ given by \eqref{eq:optimalPort}.
\end{theorem}
\begin{proof}
See Appendix \ref{app:Verify}
\end{proof}
\begin{remark}
	Herein, we take the stochastic control approach to solve the optimal investment problem. Alternatively, one may use the duality approach which, since the market model is complete, boils down to the so-called ``\emph{martingale method}'' of \cite{karatzasEtAl1987} and  \cite{cox1989}. However, since the investor's opportunity set is stochastic, one has to check extra conditions regarding the finiteness of the value function and the moments of the state price density, see \cite{cox1991} and \cite{dybvig1999portfolio}. These extra steps makes the arguments of the martingale method as complex as the one for the stochastic control approach.
\end{remark}

\section{Well-posedness condition and optimality of M-N pairs-trading}\label{sec:WellPosed}
Our findings in the previous section reveal two unsatisfactory characteristic of the investment model of Section \ref{sec:Market}. Firstly, the investment model cannot represent the \emph{equilibrium} price of a traded asset if $\gamma_0>0$.\footnote{Note that the optimal investment problem in Section \ref{sec:Optimal} is a ``\emph{partial equilibrium}'' model, which implies that $(S^1_t,S^2_t)$ corresponds to assets that are traded in a market in its ``\emph{equilibrium}'' state.} Indeed, CRRA investors with $\gamma\in(0,\gamma_0)$ achieve infinite expected utility in the finite investment horizon $T_N(\gamma)$ given by \eqref{eq:TN}. If $\gamma_0>0$, then $T_N(\gamma)\to0$ as $\gamma\to0$. Thus, one can always find CRRA investors that achieve infinite expected utility, regardless of how short the investment horizon is. As pointed out by \cite{KimOmberg1996} and \cite{KornKraft2004}, such investors cannot exist in a partial equilibrium model, since they would have pushed the market out of equilibrium by aggressively investing in the assets. 

Secondly, as one can easily check, the optimal investment strategy $\big(\pi^*(Z_t,t)\big)_{t\ge0}$ does not satisfy \eqref{ch1eq:MN} and, thus, is not M-N. This, contradicts the industry practice as explained in Section \ref{sec:MN}.

To exclude these unsatisfactory characteristic of the investment model, we should restrict the market parameters such that $\gamma_0=0$ and that the optimal strategy satisfy \eqref{ch1eq:MN}. Surprisingly, these two seemingly different requirements lead to the same condition, which we introduce next.

\begin{condition}[\emph{well-posedness}]\label{ch1cond:MNcondition}
The following equivalent relationships hold between the market parameters:
\begin{enumerate}
  \item[(i)] $\alpha_1 / \alpha_2 = (\sigma_1^2 - c \sigma_1 \sigma_2 \rho) / (\sigma_1 \sigma_2 \rho - c \sigma_2^2)$.
  
  \item[(ii)] There exists $\xi\in\mathbb{R}$ such that $\alpha = \Sigma \Sigma^\top (1,-c)^\top \xi$.
  
  \item[(iii)] $\alpha = \Sigma \Sigma^\top (1,-c)^\top (-\kappa / \sigma_Z^2)$.
\end{enumerate}
\end{condition}
\vskip1ex

\begin{remark}
	the relationships $(i) \Leftrightarrow (ii)$ and $(iii) \Rightarrow (ii)$ are trivial. To see $(ii)\Rightarrow(iii)$, left-multiply (ii) by $\left(1,-c\right)$ to obtain
	\[
	-\kappa = (1,-c)\alpha = (1,-c)\Sigma\Sigma^\top (1,-c)^\top \xi = \sigma_Z^2 \xi,
	\]
	which yields $\xi = (-\kappa/\sigma_Z^2)$.
\end{remark}

The following theorem is the main result of this paper. It characterizes the central role of Condition \ref{ch1cond:MNcondition} as the condition needed for well-posedness of the optimal investment model as well as market-neutrality of the optimal strategy.

\begin{theorem}\label{ch1thm:MNcondition}
Condition \ref{ch1cond:MNcondition} is equivalent to either of the following statements.
\begin{enumerate}
\item[(i)] For all $\gamma \in (0,1)$, the optimal investment problem (\ref{eq:Merton}) is well-posed, i.e. the maximal expected utility is finite for all investment horizon $T>0$.

\item[(ii)] The optimal strategy $\big(\pi^*(Z_t,t)\big)_{t\in[0,t]}$  is M-N, where $\pi^*$ is given by \eqref{eq:optimalPort}.
\end{enumerate}
\end{theorem}

\begin{proof}
	The equivalence with (ii) is straightforward, since $\big(\pi^*(Z_t,t)\big)_{t\in[0,t]}$ satisfies \eqref{ch1eq:MN} if and only if \ref{ch1cond:MNcondition}.(ii) holds. To show the equivalence with (i), note that by Theorem \ref{thm:Verify}, the optimal investment problem is well posed for all $\gamma\in(0,1)$ if and only if $\gamma_0=0$. From \eqref{ch1eq:gamma0}, it follows that $\gamma_0=0$ if and only if
\[
(1,-c) \Sigma \lambda = \|(1,-c) \Sigma\| \|\lambda\|.
\]
This equation is equivalent to the linear dependence of $\Sigma^\top (1,-c)^\top$ and $\lambda$ which is, in turn, equivalent to Condition \ref{ch1cond:MNcondition}.(ii).
\end{proof}
\vskip1ex

Theorem \ref{ch1thm:MNcondition} provides economic viability for the assumption that the optimal pairs-trading strategy is M-N. Indeed, real investors neither attain infinite expected utility nor take infinite positions. The implications of the possibility of attaining infinite expected utility are therefore that the parameter combinations producing such a scenario do not occur in the real world. This means that either
\begin{itemize}
	\item[(i)] $\gamma_0>0$ and there is no investor with $\gamma<\gamma_0$; or,
	
	\item[(ii)] Condition \ref{ch1cond:MNcondition} holds (i.e. $\gamma_0=0$).
\end{itemize}
Since investor's with $\gamma<\gamma_0$ are risk-averse agent's, there is no strong reason to exclude them. It then follows that Condition \ref{ch1cond:MNcondition} must hold, which, in turn, implies that the optimal investment strategy is M-N.

We end the paper by a brief discussion on the optimal strategies under Condition \ref{ch1cond:MNcondition}. By Theorem \ref{ch1thm:MNcondition}, the Merton problem is always well-posed, and imposing Condition \ref{ch1cond:MNcondition}.(iii) on \eqref{eq:optimalPort} yields the following result.

\begin{proposition}\label{ch1prop:Optimal_MN}
Under Condition \ref{ch1cond:MNcondition}, the optimal strategies is given by:
\begin{equation}\label{ch1eq:OptimalPortfolio_MN}
\pi^\star_t =
\left(\frac{-\kappa}{\sigma_Z^2}\right) \frac{1 + 1/\sqrt{\gamma} \coth\left( \frac{\kappa} {\sqrt{\gamma}}(T-t)\right)} {1 + \sqrt{\gamma} \coth\left( \frac{\kappa} {\sqrt{\gamma}}(T-t)\right)} Z_t
\begin{pmatrix}
 1\\
 -c
 \end{pmatrix};\quad t\in[0,T],
\end{equation}
for all $(\gamma,T) \in (0,1)\times (0,+\infty)$.
\end{proposition}

The form of the optimal strategy (\ref{ch1eq:OptimalPortfolio_MN}) is quite intuitive. Note that $(Z_t)$ quantify the relative mispricing between $S^1$ and $S^2$. In particular, assuming $c>0$, if $Z_t>0$ (resp. $Z_t<0$), then $S^1$ (resp. $S^2$) is over priced relative to the other stock. Since
\[
\left(\frac{-\kappa}{\sigma_Z^2}\right) \frac{1 + 1/\sqrt{\gamma} \coth\left( \frac{\kappa} {\sqrt{\gamma}}(T-t)\right)} {1 + \sqrt{\gamma} \coth\left( \frac{\kappa} {\sqrt{\gamma}}(T-t)\right)} < 0,
\]
the optimal strategy always shorts the over-priced stock and longs the under-priced one. Furthermore, the factor $-\kappa/\sigma_{z}^{2}$ tells us that the long-short positions should be bigger if the mean-reversion rate $\kappa$ is bigger, and they should be smaller if the variance rate of the residual, $\sigma_{z}^{2}$, is larger.

\section{Conclusion}\label{sec:conclude}
We considered the problem of optimal investment in a market with two cointegrated risky assets, with the motivation of finding a theoretical ground for market-neutrality of the so-called pairs-trading strategies. For this, we formulated the classical Merton problem of expected utility of terminal wealth and investigated whether this model supports, in terms of optimal choice, market-neutral pairs-trading strategies. 

We focused on the class of CRRA utilities, a model that has been studied by \cite{LiuTimmermann2012}. We found that such models might have abnormal properties, that is for some values of the risk aversion parameter, the investor attains infinite expected utility in finite investment horizon. Since such investors cannot exist in a partial equilibrium model, we identified an extra condition on the market coefficients, i.e. Condition \ref{ch1cond:MNcondition}, which eliminates the possibility of attaining infinite expected utility. Finally, we showed that Condition \ref{ch1cond:MNcondition} is equivalent to assuming that the optimal strategy is market-neutral and, hence, achieved our main goal of providing theoretical justification for the investment practice of market-neutral pairs-trading.  

\appendix

\section{Auxiliary PDE}\label{ch1app:2ndPDE}
This section provides the explicit solutions for the auxiliary Cauchy problem
\begin{equation}\label{ch1eq:2ndPDE}
\begin{cases}	
\varphi_t + z (\mathbf{a}\cdot\mathbf{b}) \varphi_z + \frac{1}{2}
\|\mathbf{b}\|^2 \varphi_{zz} + \frac{\xi}{2} z^2 \|\mathbf{a}\|^2 \varphi = 0;
&\quad (z,t) \in \R\times[0,T),\\
\varphi(z,T)=1;&\quad z\in\R,
\end{cases}
\end{equation}
and the related Riccati differential equation
\begin{equation}\label{ch1eq:Riccati}
\begin{cases}
	h^{\prime}(t) = 2 (\mathbf{a} \cdot \mathbf{b}) h(t) + \|\mathbf{b}\|^2 h^2(t)
	+ \xi \|\mathbf{a}\|^2; \quad t\in [0,T),&\quad t\in(0,T]\\
	h(0)=0.
\end{cases}
\end{equation}
It is assumed throughout that $\mathbf{a}, \mathbf{b} \in \mathbb{R}^2$, $\mathbf{a}\cdot\mathbf{b} <0$, $\xi \in
\mathbb{R}\backslash \{0\}$, and $T>0$.

Following the terminology of \cite{Sasagawa1982}, we define the ``\emph{escape criterion discriminant}''
\begin{equation}\label{ch1eq:EscapeDisc}
\mathfrak{D} := (\mathbf{a}\cdot\mathbf{b})^2 - \xi
\|\mathbf{a}\|^2 \|\mathbf{b}\|^2,
\end{equation}
and the ``\emph{escape time}'' $T_{esc}\in(0,+\infty]$ by
\begin{equation}\label{ch1eq:escapeTime}
    T_{\text{esc}}:=
	\begin{cases}
		+\infty;	&\quad\D\ge0,\\
	\frac{1}{\sqrt{-\mathfrak{D}}} \bigg( \frac{\pi}{2}
    + \arctan \Big( \frac{-\mathbf{a} \cdot
    \mathbf{b}}{\sqrt{-\mathfrak{D}}} \Big) \bigg);
	&\quad \D<0,	
	\end{cases}
\end{equation}
For $t<T_{esc}$, we also introduce the auxiliary functions
\begin{equation}\label{eq:gAux}
 g(t) = \frac{\kappa}{2\gamma} t - \frac{1}{2}
	\begin{cases}
    \log \Big( \cosh\big( t\sqrt{|\D|} \big) +\frac{\kappa}{\gamma\sqrt{|\D|}} \sinh\big(
    t\sqrt{|\D|} \big) \Big);
    &\text{if }\gamma\ne\gamma_0,\\
    \vphantom{\Bigg()}
    \log \Big( 1 + \frac{\kappa}{\gamma} t \Big) ; &\text{if
    }\gamma=\gamma_0,
	\end{cases}
\end{equation}
and
\begin{equation}\label{eq:hAux}
    h(t) =
	\begin{cases}
    \frac{\displaystyle\vphantom{\Big(} \xi \|\mathbf{a}\|^2}
    {\displaystyle\vphantom{\bigg(} -\mathbf{a} \cdot \mathbf{b} +
    \sqrt{\mathfrak{D}} \coth \Big( t\sqrt{\mathfrak{D}} \Big) };
    &\quad \text{if }\D>0,\\
    \frac{\displaystyle\vphantom{\big(} \mathbf{a}\cdot\mathbf{b}}
    {\displaystyle\vphantom{\Big(} \|\mathbf{b}\|^2}\left(
    \frac{\displaystyle\vphantom{\big(}1}
    {\displaystyle\vphantom{\Big(} 1 - (\mathbf{a}\cdot\mathbf{b}) t}
    - 1\right); &\quad \text{if }\D=0,\\
	-\frac{\sqrt{-\mathfrak{D}}}
    {\|\mathbf{b}\|^2}\tan\bigg( \arctan \Big( \frac{-\mathbf{a} \cdot
    \mathbf{b}} {\sqrt{-\mathfrak{D}}}\Big) -
    \sqrt{-\mathfrak{D}}t \bigg) - \frac{\mathbf{a} \cdot \mathbf{b}}
    {\|\mathbf{b}\|^2};
	&\quad \text{if }\D<0.
	\end{cases}
\end{equation}

The following Lemmas provide the solutions of \eqref{ch1eq:2ndPDE} and \eqref{ch1eq:Riccati} as well as some of their properties. Note that the solutions are defined up to $T_{esc}$. In particular, for $\D<0$, the solutions ``\emph{blow up}'' at the finite escape time $T_{esc}$. The proof of the lemmas are by direct substitution and omitted.

\begin{lemma}\label{lem:RDE}
For $T<T_{esc}$, the solution of the Riccati equation \eqref{ch1eq:Riccati} is given by $h(t)$. Furthermore:
\begin{enumerate}
	\item[(i)] If $\xi > 0$ (resp. $\xi < 0$), then $h$ is positive and strictly increasing (resp. negative and strictly decreasing).
	
	\item[(ii)] If $\D\ge0$ (resp. $\D<0$), then$\displaystyle \lim_{t\to+\infty} h_(t) = \frac{ \xi \|\mathbf{a}\|^2}
     {\sqrt{\mathfrak{D}} -\mathbf{a} \cdot \mathbf{b}}$ (resp. $\displaystyle\lim_{t\to T_{esc}^-} h(t)= +\infty$).
\end{enumerate}
\end{lemma}\vspace{1em}

\begin{lemma}\label{lem:AuxPDE}
For $T<T_{esc}$, the solution of the Cauchy problem \eqref{ch1eq:2ndPDE} is given by
\begin{equation}\label{eq:AuxPDESol}
	\varphi(z,t) = \exp\Big( g(T-t) + \frac{1}{2} h(T-t) z^2\Big);\quad (z,t)\in\R\times[0,T].
\end{equation}
In particular, if $\D<0$, then $\lim_{T\to T_{esc}^-} \varphi(z,0) = +\infty$, for all $z>0$.
\end{lemma}

\section{Proof of Theorem \ref{thm:Verify}}\label{app:Verify}

	It is sufficient to show the following statements.
	\begin{enumerate}
		\item[(i)] For any admissible strategy $\pi=(\pi^1_t,\pi^2_t)\in\mathcal{A}$, one has
		\begin{equation}\label{eq:Verification}
			v(x,z,t)\ge\E^{x,z,t}\frac{(X^\pi_T)^{1-\gamma} -1}{1-\gamma};\quad (x,z,t)\in\R^+\times\R\times[0,T].
		\end{equation}
	
		\item[(ii)] $\pi^*\in\mathcal{A}$ and
		\begin{equation}\label{eq:Verification2}
			v(x,z,t) = \E^{x,z,t}\frac{(X^*_T)^{1-\gamma} -1}{1-\gamma};\quad (x,z,t)\in\R^+\times\R\times[0,T],
		\end{equation}
		where we defined $X^*=X^{\pi^*}$.
	\end{enumerate}
	\vskip 1ex
	\noindent\textbf{Proof of (i):} For $n>0$, define the stopping time
	\begin{equation}
		\tau_n := T\wedge\inf\Big\{t\ge0: \max\{\int_0^t \pi_u^2 du, \vert X^\pi_t \vert, \vert Z_t \vert\} > n \Big\}\bigg\}.
	\end{equation}
	Note that $\tau_n \to T$ a.s. as $n\to+\infty$. Applying It\^o's rule yields
	\begin{equation}\label{eq:ItoExpansion}
	\begin{split}	
		v(X^\pi_{\tau_n},Z_{\tau_n},\tau_n) &= v(t,x,z) + \int_t^{\tau_n} \mathcal{L}^{\pi_u} v(X^\pi_u,Z_u,u) du \\
		&\quad + \int_t^{\tau_n} v_z(X^\pi_u,Z_u,u) \sigma_Z  dW^Z_u + \int_t^{\tau_n} v_x(X^\pi_u,Z_u,u) X^\pi_u \pi^\top_u \Sigma dW_u.
	\end{split}
	\end{equation}
	The first integral on the right side is non-positive because $v$ solves the HJB equation \eqref{eq:HJB}. Furthermore, by the definition of $\tau_n$, the integrands of the second and third integrals are uniformly bounded, thus,
	\[
	 	\E^{x,z,t}\int_t^{\tau_n} v_z(X^\pi_u,Z_u,u) \sigma_Z  dW^Z_u= \E^{x,z,t}\int_t^{\tau_n} v_x(X^\pi_u,Z_u,u) X^\pi_u \pi^\top_u \Sigma dW_u=0.
	\]
	Therefore, taking the conditional expectation on both sides of (\ref{eq:ItoExpansion}) yields
	\begin{equation}\label{eq:ExpectedIto}
		\E^{x,z,t}v(X^\pi_{\tau_n},Z_{\tau_n},\tau_n) \le v(x,z,t);\quad n\ge1, (x,z)\in \R^+\times \R\times [0,T].
	\end{equation}
	Finally, $\gamma\in(0,1)$ yields that $v$ is uniformly bounded from bellow, specifically, $v_z(X^\pi_t,Z_t,t)\ge \frac{1}{\gamma-1}$.  
	Thus, (\ref{eq:Verification}) is obtained by taking the limit of (\ref{eq:ExpectedIto}) as $n\to+\infty$ and applying Fatou's lemma.
	
	\vskip 1ex
	\noindent\textbf{Proof of (ii):} To show that $\pi^*\in\mathcal{A}$, it suffices to check the integrability condition \eqref{ch1eq:integrability} for $\big(\pi^*(Z_t,t)\big)_{0\le t\le T}$. From \eqref{eq:optimalPort}, it follows that
	\[
	\begin{split}
		|\pi_t^{*\top} \alpha Z_t| + \pi_t^{*\top} \Sigma\Sigma^\top \pi_t^* =  Z_t^2 \; \bigg\{ &\Big|\frac{1}{\gamma} \alpha^\top(\Sigma\Sigma^\top)^{-1}\alpha - \kappa h(T-t)\Big| \\
		&+ \frac{1}{\gamma^2} \alpha^\top (\Sigma\Sigma^\top)^{-1}\alpha + \sigma_Z^2 h^2(T-t) - \frac{2\kappa}{\gamma}h(T-t)\bigg\}.
	 \end{split}
	\]
	Since $T< T_N(\gamma)$, Lemma \ref{lem:Monotone} yields that $|h(T-t)|$ is uniformly bounded for $t\in[0,T]$. Furthermore, by \eqref{ch1eq:Z_moments}, $\E\int_0^T Z_u^2 du = \frac{\sigma_Z^2T}{2\kappa}<+\infty$. Thus,
	\[
		\E \int_0^T |\pi_t^{*\top} \alpha Z_t| + \pi_t^{*\top} \Sigma\Sigma^\top \pi_t^* dt <+\infty,
	\]
	and $\pi^*$ satisfies \eqref{ch1eq:integrability}.
	
	To prove \eqref{eq:Verification2}, we repeat the proof of (i) for $\pi=\pi^*$ to obtain
	\begin{equation}\label{eq:ExpectedIto2}
		\E^{x,z,t}v(X^*_{\tau_n},Z_{\tau_n},\tau_n) = v(x,z,t);\quad n\ge1, (x,z)\in \R^+\times \R\times [0,T].
	\end{equation}
	If there exists a constant $p>1$ such that
	\begin{equation}\label{eq:UI}
		\sup_{0\le t\le T}\E\big|v(X^*_t,Z_t,t)\big|^p <+\infty,
	\end{equation}
	then the sequence $\{v(X^*_{\tau_n},Z_{\tau_n},\tau_n)\}_{n=1}^{+\infty}$ is uniformly integrable and \eqref{eq:Verification2} is obtained by letting $n\to+\infty$ in \eqref{eq:ExpectedIto2}.
	
	It only remain to show \eqref{eq:UI}. For the rest of the proof, let $p>1$ be an arbitrary constant (say, $p=2$). Since $T<T_N(\gamma)$, the functions $g$ and $h$ are uniformly bounded on $[0,T]$. Thus, there exists constants $K_1,K_2>0$, such that, for all $t\in[0,T]$, 
	\[
	\begin{split}
		\big|v(X^*_t,Z_t,t)\big|^p &< \frac{(X^*_t)^{p(1-\gamma)}}{(1-\gamma)^p} \left(e^{g(T-t)+\frac{1}{2}h(T-t)Z_t^2}\right)^{p\gamma}\\
		&< K_1 \exp\Big(K_2\int_0^t Z_u du\Big)\;\e\Big(\int_0^. Z_u \Big[\frac{1}{\gamma}\alpha^\top\Sigma^{-1\top}+(h(T-u)+p\gamma)(1,-c)\Sigma\Big]dW_u\Big)_t.
	\end{split}
	\]
	Here, $\e(Y)_t:=\exp(\int_0^t dY_t -0.5\int_0^t d[Y_t])$ denotes the Dol\'eans-Dade exponential of process $(Y_t)$. By Lemma \ref{lem:StationarityOfResidual}, $\int_0^t Z_u du$ has the Gaussian distribution with mean $0$ and
	\[
		\E\big[\int_0^t Z_u du\big]^2 \le \E\big[\int_0^T Z_u^2 du\big] = \frac{\sigma^2_Z T}{2\kappa}.
	\]
	Therefore,
	\begin{equation}\label{eq:Sup1}
		\sup_{0\le t\le T} \E\exp\Big(K_2\int_0^t Z_u du\Big) = \sup_{0\le t\le T} \exp\Big(\frac{1}{2} K_2^2 \E\big[\int_0^t Z_u du\big]^2\Big)\le \exp\Big(\frac{K_2^2\sigma^2_Z T}{4\kappa}\Big)<+\infty.
	\end{equation}
	Furthermore by Corollary 11 on page 85 of \cite{Krylov1980},
	\begin{equation}\label{eq:Sup2}
	\begin{split}
		\sup_{0\le t\le T} \E \e\Big(\int_0^. Z_u &\Big[\frac{1}{\gamma}\alpha^\top\Sigma^{-1\top}
		+(h(T-u)+p\gamma)(1,-c)\Sigma\Big]dW_u\Big)_t\\
		&\le \E \sup_{0\le t\le T} \e\Big(\int_0^. Z_u \Big[\frac{1}{\gamma}\alpha^\top\Sigma^{-1\top}
		+(h(T-u)+p\gamma)(1,-c)\Sigma\Big]dW_u\Big)_t<+\infty.
	\end{split}
	\end{equation}
	Finally, we obtain \eqref{eq:UI} as follows,
	\[
	\begin{split}
		\sup_{0\le t\le T}\E\big|v(X^*_t,Z_t,t)\big|^p \le &\sup_{0\le t\le T} \E\exp\Big(K_2\int_0^t Z_u du\Big)\\
		&\times \sup_{0\le t\le T} \E \e\Big(\int_0^. Z_u \Big[\frac{1}{\gamma}\alpha^\top\Sigma^{-1\top}
		+(h(T-u)+p\gamma)(1,-c)\Sigma\Big]dW_u\Big)_t<+\infty.
	\end{split}
	\]


\end{document}